\documentclass{article}
\usepackage[utf8]{inputenc}
\usepackage{amsmath,amssymb,amsthm,bm}
\usepackage{authblk}
\usepackage{natbib}
\usepackage{graphicx}
\usepackage{xcolor}
\usepackage[margin=1in]{geometry}
\newcommand{\cov}{\textnormal{Cov}}
\newcommand{\cor}{\textnormal{Cor}}
\newcommand{\var}{\textnormal{Var}}

\newcommand{\e}{\mathbb{E}}
\newtheorem{proposition}{Proposition}
\usepackage[normalem]{ulem}

\theoremstyle{remark}
\newtheorem{remark}{Remark}
\newtheorem{assumption}{Assumption}

\title{Setting the duration of online A/B experiments}
\author[1]{Harrison H. Li }
\author[2]{Chaoyu Yu \footnote{Corresponding author. Email: chaoyuyu@google.com}}
\affil[1]{Department of Statistics, Stanford University}
\affil[2]{Google Inc.}
\date{August 2024}

\begin{document}

\maketitle

\begin{abstract}
In designing an online A/B experiment, it is crucial to select a sample size and duration that ensure the resulting confidence interval (CI) for the treatment effect is the right width to detect an effect of meaningful magnitude with sufficient statistical power without wasting resources.
While the relationship between sample size and CI width
is well understood, the effect of experiment duration on CI width remains less clear.  
This paper provides an analytical formula for the width of a CI based on a ratio treatment effect estimator as a function of both sample size ($N$) and duration ($T$).
The formula is derived from a mixed effects model with two variance components.
One component, referred to as the \textit{temporal variance}, persists over time for experiments where the same users are kept in the same experiment arm across different days. The remaining \textit{error variance} component, by contrast, decays to zero as $T$ gets large. The formula we derive introduces a key parameter that we call the user-specific temporal correlation (UTC), which quantifies the relative sizes of the two variance components and can be estimated from historical experiments. Higher UTC indicates a slower decay in CI width over time.
On the other hand, when the UTC is 0 —-- as for experiments where users shuffle in and out of the experiment across days —-- the CI width decays at the standard parametric $1/T$ rate. We also study how access to pre-period data for the users in the experiment affects the CI width decay.
We show our formula closely explains CI widths on real A/B experiments at YouTube.
\end{abstract}

\section{Introduction}
A typical A/B test (randomized experiment) to evaluate the impact of a proposed feature in a large-scale online setting 
randomly diverts a pre-specified proportion of all user traffic (e.g. 1\%) into an experiment.
Half of these diverted users are randomly assigned into the control arm, 
and do not see any change from the status quo.
The remaining diverted users form the treatment arm and do see the proposed change. 

In many cases, an experimenter would like the same users to be tracked over the course of the different days of the experiment.
This requires a user who was in a particular arm of the experiment on a previous day
to remain in this arm on all subsequent days of the experiment. At YouTube, we call this type of experiment the \textit{user experiment}, to be contrasted with the \textit{user-day experiment} where we get a fresh set of users for each arm daily.
One reason to run a user experiment is that randomly shuffling a user in and out of an experiment may create a ``disorienting" or otherwise undesirable experience~\citep{tang2010layer}.
Another reason is illustrated by an experiment to measure whether a personalized website homepage increases the number of daily visitors.
In that setting we may want to understand how exposure to the new homepage over time impacts key user engagement metrics, requiring us to keep each user within the same experiment arm for weeks or months \citep{hohnhold2015focusing}. In this paper, we mostly focus on the user experiment;
in Section~\ref{sec:user-day},
we show how the user-day experiment arises as a special limiting case of our analysis.

When designing an experiment, investigators need to ask themselves the following two questions:
\begin{itemize}
\item How long ($T$) should the experiment run? A day? A week? Many weeks?
\item How large ($N$) should the experiment be? 1\%, 10\%, or 50\% of all users?
\end{itemize}
The practice of deciding the sample size ($N$) of the experiment, referred to as ``experiment sizing", is widely viewed as a crucial step in designing experiments~\citep{Kohavi2022ab}. 
A typical experiment sizing workflow chooses $N$ to achieve some desired level of statistical power to detect a meaningful change in a metric of interest via a sufficiently narrow CI~\citep{van2011statistical}.
However, this power, which is determined by the variance of the treatment effect estimate, depends on both the sample size and the length of the experiment. Holding the experiment duration constant, the relationship between the variance of a standard treatment effect estimate and the sample size $N$ is clear: assuming independence across users, standard textbook results indicate the variance is proportional to $1/N$,
and hence the width of the corresponding CI should be proportional to $1/\sqrt{N}$. 
% But the relationship between CI width and the experiment duration $T$ for fixed $N$ is less clear. 
%A common misconception is that by running the experiment longer, we can always reduce the variance of the metrics at the same rate. 

At YouTube,
however,
we observe that keeping the sample size $N$ fixed and increasing the duration $T$,
the CI width often decays much more slowly than $1/\sqrt{T}$. Fig.~\ref{fig:size_length_experiment} shows CI widths as a function of $T$,
averaged over 500 YouTube experiments for four selected metrics:
\begin{itemize}
    \item DAU (Daily Active Users): The total count of users who are active on YouTube on a given day.
    \item Views: The total views of all YouTube videos.
    \item CTR (Click Through Rate): The click through rate of search results on YouTube.
    \item Crashes: The total number of YouTube sessions with crashes.
\end{itemize}
The CI width for each metric and experiment is normalized (divided) by the day 1 CI width to illustrate the decay in the CI width over time. The black solid curve in Fig.~\ref{fig:size_length_experiment} describes the decay rate $1/\sqrt{T}$. As we can see, different metrics showed quite different rates of CI width shrinkage over time. After 14 days, we find that the CI width for DAU is, on average, only 24$\%$ smaller than the CI width from running the experiment for only one day. On the other hand, for the Crashes metric, we can reduce the CI width by 58$\%$ by running the experiment for an extra 13 days. As a reference, if we received a new, independent sample of users each day, 13 extra days of the experiment would shrink the CI width by $100 \times (1-1/\sqrt{14}) \% \approx 73\%$.

\begin{figure}
    \centering
    \includegraphics[width=0.8\linewidth]{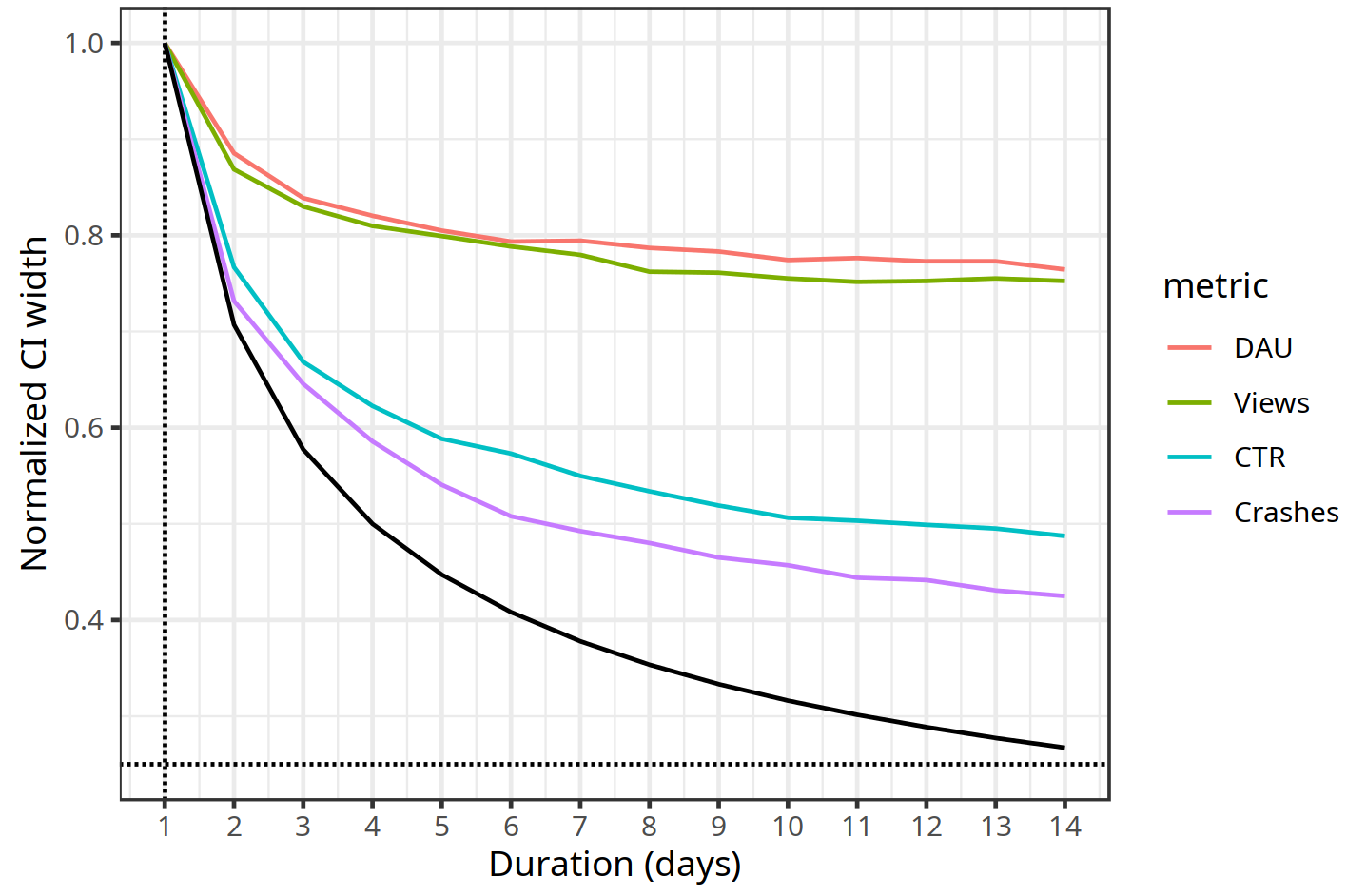}
    \caption{The relationship between the width of the CI for the treatment effect (normalized by the day 1 CI width) and the experiment duration $T$. The black solid curve corresponds to the $1/\sqrt{T}$ decay rate.}
    \label{fig:size_length_experiment}
\end{figure}

In this paper, we provide an analytical formula that models the variance of a ratio treatment effect estimator as a function of the duration $T$.
This provides a structural explanation of the shapes of the curves shown in Fig.~\ref{fig:size_length_experiment}. 
The formula,
stated in Section~\ref{sec:main_result},
depends on a key parameter called the user-specific temporal correlation (UTC) that can vary for different metrics.
In Section~\ref{sec:model_detail},
we show how our formula can be derived from a two component mixed effects model that attempts to explain the correlation structure across different days of an experiment.
Then we discuss how access to pre-period data for the users in the experiment affects our variance calculations (Section~\ref{sec:prepost})
and introduce user-day experiments as a special case of the user experiment (Section~\ref{sec:user-day}).
Finally, we show how to estimate the UTC using prior experiments and perform a numerical study on real YouTube experiments in Section~\ref{sec:estimation}.
The overall aim is to facilitate a power analysis that provides guidance on experiment sizing and duration,
enabling better resource management and developer velocity.

\section{CI width as a function of experiment duration}
\label{sec:main_result}
In this section, we present the main result of this paper: a simple formula that describes the relationship between CI width and experiment duration. \\

\noindent
\textbf{Main result}: Let $|CI(T;N)|$ represent the CI width for the ratio (percentage) treatment effect of a metric in an A/B experiment with duration $T$ and sample size $N$.
Then for $N$ sufficiently large,
we expect
\begin{equation}
    \frac{|CI(T;N)|}{|CI(1;N)|} \approx \sqrt{\frac{1}{T} + \rho \frac{T-1}{T}}
    \label{eq:ci_decay_rate}
\end{equation}
where $\rho$ is the user-specific temporal correlation (UTC)
and is a metric-specific parameter that describes how much a metric is correlated across different experiment days due to the persistence of the same users in the experiment.

\begin{remark}
We assume $\rho \in [0, 1]$ throughout. 
See Section~\ref{sec:estimation} for a discussion of different strategies for estimating $\rho$ using previous A/B experiments. 
\end{remark}

\begin{remark} Given that $\sqrt{\frac{1}{T} + \rho \frac{T-1}{T}} \geq 1/\sqrt{T}$ for all $\rho \in [0,1]$,~\eqref{eq:ci_decay_rate} shows that multiplying the traffic size $N$ of an A/B experiment by $c \geq 1$ increases statistical power by no less than running the A/B experiment $c$ times longer.
\end{remark}

The right-hand side of~\eqref{eq:ci_decay_rate} is an increasing function of $\rho \in [0, 1]$: the higher the user-specific temporal correlation, 
the more slowly the CI will shrink over time. This is consistent with our findings in Fig.~\ref{fig:size_length_experiment}. Metrics like DAU are highly correlated over time: a user who showed up on YouTube every day last week is much more likely to show up on YouTube tomorrow than a user who did not up on YouTube in the past week. As a result, the UTC $\rho$ is high and running the experiment longer does not help narrow the CI of DAU much. On the other hand, for a metric like Crashes, the UTC is smaller and the CI width decay is closer to the fast $1/\sqrt{T}$ rate.

We note in particular that~\eqref{eq:ci_decay_rate} implies
\begin{equation}
\label{eq:positive}
\lim_{T \rightarrow \infty} |CI(T;N)| \approx \sqrt{\rho}|CI(1;N)|.
\end{equation}
Thus if $\rho>0$,
no matter how long we run the experiment,
the CI width will not shrink to zero.
The statistical intuition behind this phenomenon is that with the same users appear in the experiment across different days,
we can only possibly obtain information from a finite subset of users in the superpopulation,
even if we run the experiment infinitely long.

If an investigator additionally has access to pre-experiment data on the same users,
Section~\ref{sec:prepost} derives the following modification to~\eqref{eq:ci_decay_rate}:
\begin{equation}
\label{eq:ci_decay_rate_PP}
   \frac{|CI(T;N)|}{|CI(1;N)|} \approx \sqrt{\frac{1}{T} + \rho \frac{T-1}{T}} \sqrt{\frac{1-\lambda(T)^2}{1-\lambda(1)^2}}
\end{equation}
where
\[
\lambda(T) = \frac{1}{\sqrt{1+\frac{1-\rho}{\rho T_0}}\sqrt{1+\frac{1-\rho}{\rho T}}} 
\]
for $T_0$ the number of days of pre-experiment data available for each user.
For $\rho \in (0,1)$,
we have $\lambda(T)$ is increasing in $T$ and hence
access to pre-experiment data enables a faster decay in the CI width with time
(in addition to narrower CI's overall for all choices of $T$).

\section{A two-component mixed effects model for A/B experiment metrics}
\label{sec:model_detail}
In this section, we derive our formula~\eqref{eq:ci_decay_rate}
under various modeling assumptions that we state explicitly.
A user not interested in the mathematical details is encouraged to skip to Section~\ref{sec:estimation} for a discussion on how to estimate the UTC parameter $\rho$ and a validation of~\eqref{eq:ci_decay_rate} on real YouTube experiments.

We will derive~\eqref{eq:ci_decay_rate} for two types of metrics that differ depending on how they aggregate across users: additive metrics and ratio metrics.
Additive metrics like DAU are computed by simply summing across users,
while ratio metrics like CTR are computed as ratios of the sums of two additive metrics 
(e.g. clicks and impressions) across users. 
% We will start by modeling experiment metrics on the user level and then aggregate over users in each arm accordingly.
% The result is a two-component mixed effects model.
% We then derive the metric variance as a function of both sample size $N$ and duration $T$. 
% Afterwards, we will discuss the application of our framework on a CI procedure where pre-experiment metrics are used for variance reduction. Finally, we briefly discuss the special case of user-day experiments where users are shuffled in and out of the experiment arm on a daily basis.

\subsection{User level modeling}
\label{sec:user}
We model user level behavior with the following common assumption,
which precludes the existence of network effects (see~\citet{xu2015network} for a further discussion of network effects):

\begin{assumption}
\label{assump:iid_users}
Users are sampled independently from a single superpopulation.
\end{assumption}

Mathematically, Assumption~\ref{assump:iid_users} ensures that measurements from different users are independent and identically distributed (i.i.d.).
It is reasonable for an experiment testing a feature where there is limited interaction between users.

We let $a_{utj}$ denote the value of a generic additive metric for user $u=1,\ldots,N$ in experiment arm $j \in \{r,c\}$ on day $t = 1,\ldots,T$  
(here $r$ denotes the treatment arm and $c$ denotes the control arm).
Note $N$ denotes the size of each arm,
so in our notation,
the total number of users in the experiment is $2N$.
A generic ratio metric is the ratio of two additive metrics,
which we call the ``numerator" metric and the ``denominator" metric.
We denote their values for user $u$ in experiment arm $j$ on day $t$ by $a_{utj;n}$ and $a_{utj;d}$, respectively.
Throughout the remainder of this paper,
all notation and assumptions pertaining to our generic additive metric $a_{utj}$ will carry over without further comment to the numerator and denominator metrics for our generic ratio metric (appending $;n$ and $;d$, as appropriate).
Table~\ref{table:notation} collects all such notation.

% Define $m_{utj} = \e[a_{utj}]$ to
% be the mean metric value for user $u$
% in arm $j$ on day $t$ of the experiment
% (averaging over the randomness in the user sampling),
% so that we can write
% \begin{equation}
% a_{utj} = m_{utj} + e_{utj} \label{eq:user_decomposition} \\
% \end{equation}
% where $e_{utj}$ has mean 0 and some (finite) variance $\sigma_E^2 > 0$.
% All notation is collected in Table~\ref{table:notation} below.
Our key modeling assumption is the following mixed effects model for $a_{utj}$:

\begin{assumption}
\label{assump:user_mean_decomposition}
For each user $u$, day $t$, and arm $j$, we have
\begin{equation}
a_{utj} = m_{tj} + \alpha_{uj} + e_{utj} \label{eq:user_mean_decomposition}.
\end{equation}
for fixed effects $m_{tj}$ and random effects $\alpha_{uj}$ independent of the errors $e_{utj}=a_{utj}-m_{utj}$.
We further assume the collection $\{\alpha_{uj} \mid u=1,\ldots,N, j=r,c\}$ is i.i.d. with mean 0 and variance $\sigma_A^2 < \infty$,
and that the collection of vectors $\{(e_{u1j},\ldots,e_{uTj}) \mid u=1,\ldots,N, j=r,c\}$ is i.i.d. with $\e(e_{utj})=0$ and $\var(e_{utj}) = \sigma_E^2 < \infty$ for each user $u=1,\ldots,N$, day $t=1,\ldots,T$, and arm $j=r,c$.
\end{assumption}

In~\eqref{eq:user_mean_decomposition},
$m_{tj}$
corresponds to the mean metric value in the superpopulation of users on day $t$ in arm $j$,
and is viewed as non-random.
% we instead view $m_{tj}$ as a fixed effect capturing the day-to-day variations in the metric means in the superpopulation.
% Since we are only concerned with modeling variances, the $m_{tj}$ are nuisance parameters and do not need to be modeled.
In many practical settings,
we find that the day-to-day variability in a given metric of interest
swamps the magnitude of the treatment effect.
Thus, it is crucial to allow the $m_{tj}$ to vary with $t$.
Note, however, that we will not need to impose any assumptions on the $m_{tj}$ to derive our CI width formula~\eqref{eq:ci_decay_rate},
so they can henceforth be viewed as nuisance parameters.
The term $\alpha_{uj}$ in~\eqref{eq:user_mean_decomposition} is an additive random effect for user $u$ in arm $j$,
designed to capture the idiosyncracies in the behavior of the particular user $u$ in treatment arm $j$.
The main assumption captured by~\eqref{eq:user_mean_decomposition} is that such idiosyncracies are time invariant,
highlighted by the lack of a time subscript on these random effects.

\begin{table}[!htb]
\centering
\vspace{0.2in}
\label{table:notation}
\begin{tabular}{|p{0.2\linewidth}|p{0.7\linewidth}|}
\hline
Quantity & Description \\
\hline
$a_{utj}$ & Metric value for user $u$ on day $t$ in arm $j$ \\
\hline
$m_{tj}$ & Mean metric value on day $t$ in the superpopulation of users in arm $j$ \\
\hline
$\alpha_{uj}$ & Additive random effect for user $u$ in arm $j$ \\
\hline
$\sigma_A$ & Standard deviation of $a_{uj}$ \\
\hline
$e_{utj}$ & Randomness in user $u$'s metric value on day $t$ in arm $j$, i.e. $e_{utj} = a_{utj}-m_{utj}$ \\
\hline
$\sigma_E$ & Standard deviation of $e_{utj}$ \\
\hline
$\bar{\sigma}_E(T)$ & Standard deviation of time-averaged errors $T^{-1} \sum_{t=1}^T e_{utj}$ \\
\hline
$\sigma_{A;nd}$ & For ratio metrics only; the covariance between the random effects for the numerator and denominator metrics, i.e. $\cov(\alpha_{uj;n},\alpha_{uj;d})$ \\
\hline
$\bar{\sigma}_{E;nd}(T)$ & For ratio metrics only; the covariance between the time-averaged random fluctuations of the numerator and denominator metrics,
i.e. $\cov\left(T^{-1} \sum_{t=1}^T e_{utj;n}, T^{-1} \sum_{t=1}^T e_{utj;d}\right)$ \\
\hline
\end{tabular}
\caption{Some notation used throughout the text.
}
\end{table}

% We assume these are each generated in exactly the same way as the $a_{utj}$ above.
% That is, $n_{utj}=m_{utj;n}+e_{utj;n}$ for $e_{utj;n}  \sim  [0,\sigma_{E;n}^2]$ 
% and $d_{utj}=m_{utj;n}+e_{utj;d}$ for $e_{utj;d} \sim [0,\sigma_{E;d}^2]$,
% with $m_{utj;n}=m_{tj;n} + a_{uj;n}$ 
% and $m_{utj;d}=m_{tj;d} + a_{uj;d}$ 
% for additive random effects $a_{uj;n} \sim [0,\sigma_{A;n}^2]$ 
% and $a_{uj;d} \sim  [0,\sigma_{A;d}^2]$
% jointly independent of the errors $(e_{utj;d},e_{utj;n})$. 
% Note we allow for the numerator and denominator random effects and errors to be correlated.
% Then we let $\sigma_{A;nd}=\cov(a_{uj;n},a_{uj;d})$ be the covariance between the numerator and denominator random effects
% and define $\sigma_{E;nd}=\cov(e_{utj;n},e_{utj;d})$ to be the covariance between the numerator and denominator errors (on the same day of the experiment).

\subsection{Ratio treatment effect}

The estimand of interest 
is the ratio treatment effect $\theta$,
given by
$\theta = m_{tr}/m_{tc} - 1$ for the additive metric and  
$\theta = (m_{tr;n}/m_{tr;d})/(m_{tc;n}/m_{tc;d}) - 1$ 
for the ratio metric. The scale invariant nature of this quantity makes it a ubiquitous choice of estimand in online experimentation, enabling direct comparisons of treatment effects across diverse experiments and metrics. While we focus on the ratio treatment effect in our derivation of~\eqref{eq:ci_decay_rate},
we note that our arguments can be adapted to show that~\eqref{eq:ci_decay_rate} holds for estimating additive treatment effects as well, under the same assumptions.

% We note that for the ratio metric,
% $\theta$ does not pertain to difference between the mean values of the ratios $m_{tr;n}/m_{tr;d}$ and $m_{tc;n}/m_{tc;d}$.
% Rather, it corresponds to the difference between the ratio of the numerator metric mean to the denominator metric mean in treated users
% and the same ratio of means in control users due to treatment.
% These quantities can be quite different if either metric is highly skewed.
% The motivation for our definition comes from the fact that a ratio metric
% is a ratio of sums (or averages) over users rather than an average of ratios across users.
% For instance, the clickthrough rate on a video advertisement is the sum of user clicks divided by the sum of ad impressions,
% not the average click to impression ratio across users.

The lack of a subscript $t$ in the definition of $\theta$ implicitly encodes the following assumption:

\begin{assumption}[No learning effects]
\label{assump:constant_ratio}
The ratios $m_{tr}/m_{tc}$, $m_{tr;n}/m_{tc;n}$, and $m_{tr;d}/m_{tc;d}$ are independent of $t$.
\end{assumption}

Assumption~\ref{assump:constant_ratio} implies the absence of any ``learning effects" where the treatment effect varies over time.
It is a simplifying assumption
to enable the aggregation of data across different days of the experiment
to consistently estimate the same estimand $\theta$,
regardless of the chosen duration of the experiment.

\subsection{Estimator and CI construction}
Our CI's for $\theta$ are based on the natural estimator
\[
\hat{\theta} = 
\begin{cases}
\frac{\bar{a}_{\cdot \cdot r}}{ \bar{a}_{\cdot \cdot c}} - 1 & \text{ for the additive metric} \\
\frac{\bar{a}_{\cdot \cdot r;n} / \bar{a}_{\cdot \cdot r;d}}{\bar{a}_{\cdot \cdot c;n} / \bar{a}_{\cdot \cdot c;d}} - 1 & \text{ for the ratio metric}
\end{cases}
\]
Above and hereafter,
a bar over a letter
together with a dot $\cdot$ in a subscript means taking a sample average over the dimension represented by the subscript,
i.e.
\[
\bar{a}_{\cdot\cdot j} = (NT)^{-1} \sum_{u=1}^N \sum_{t=1}^T a_{utj}, \qquad j \in \{r,c\}.
\]
is the average additive metric value over all users in arm $j$ across all days $t=1,\ldots,T$ of the experiment.

We now show that the estimator $\hat{\theta}$ satisfies a central limit theorem for each fixed $T$,
letting $N \rightarrow \infty$:
\begin{proposition}
\label{prop:clt_additive}
Suppose Assumptions~\ref{assump:iid_users}-\ref{assump:constant_ratio} hold for an additive metric $a_{utj}$.
Then in the notation of Table~\ref{table:notation} we have $\sqrt{N}(\hat{\theta}-\theta) \stackrel{d}{\rightarrow} \mathcal{N}(0,V_{\hat{\theta}}(T))$
as $N \rightarrow \infty$, where
\begin{equation}
\label{eq:theta_hat_var_additive}
V_{\hat{\theta}}(T) = (\theta + 1)^2 (\sigma_A^2 + \bar{\sigma}_E^2(T))(\bar{m}_{\cdot r}^{-2} + \bar{m}_{\cdot c}^{-2}).
\end{equation}
\end{proposition}
\begin{proof}
By the standard central limit theorem, we have
\[
\sqrt{N}
\left(
\begin{bmatrix}
\bar{a}_{\cdot\cdot r} \\
\bar{a}_{\cdot\cdot c}
\end{bmatrix}
- 
\begin{bmatrix}
\bar{m}_{\cdot r} \\
\bar{m}_{\cdot c}
\end{bmatrix}
\right)
\stackrel{d}{\rightarrow}
\mathcal{N}\left(
\begin{bmatrix}
0 \\
0
\end{bmatrix},
\begin{bmatrix}
\sigma_A^2 + \bar{\sigma}_E^2(T) & 0 \\
0 & \sigma_A^2 + \bar{\sigma}_E^2(T)
\end{bmatrix}
\right)
\]
Then equation~\eqref{eq:theta_hat_var_additive} follows by the Delta method.
\end{proof}
\begin{proposition}
\label{prop:clt_ratio}
Suppose Assumptions~\ref{assump:iid_users}-\ref{assump:constant_ratio} hold for a ratio metric with numerator and denominator metrics $a_{utj;n}$ and $a_{utj;d}$,
respectively.
Then in the notation of Table~\ref{table:notation} we have
$\sqrt{N}(\hat{\theta}-\theta) \stackrel{d}{\rightarrow} \mathcal{N}(0,V_{\hat{\theta}}(T))$
as $N \rightarrow \infty$, where
\begin{align}
V_{\hat{\theta}}(T) = (\theta+1)^2\sum_{j \in \{r,c\}} &\Bigg[ \frac{\sigma_{A;n}^2 + \bar{\sigma}_{E;n}^2(T)}{\bar{m}_{\cdot j;n}^2} + \frac{\sigma_{A;d}^2 + \bar{\sigma}_{E;d}^2(T)}{\bar{m}_{\cdot j;d}^2} -  \frac{2(\sigma_{A;nd} +   \bar{\sigma}_{E;nd}(T))}{\bar{m}_{\cdot j;n} \bar{m}_{\cdot j;d}}\Bigg] \label{eq:theta_hat_var_ratio}
\end{align}
\end{proposition}
\begin{proof}
For each arm $j \in \{r,c\}$,
by the standard central limit theorem we have
\[
\sqrt{N}
\left(
\begin{bmatrix}
\bar{a}_{\cdot\cdot j;n} \\
\bar{a}_{\cdot\cdot j;d}
\end{bmatrix}
- 
\begin{bmatrix}
\bar{m}_{\cdot\cdot j;n} \\
\bar{m}_{\cdot\cdot j;d}
\end{bmatrix}
\right)
\stackrel{d}{\rightarrow}
\mathcal{N}\left(
\begin{bmatrix}
0 \\
0
\end{bmatrix},
\begin{bmatrix}
\sigma_{A;n}^2 + \bar{\sigma}_{E;n}^2(T) & \sigma_{A;nd} + \bar{\sigma}_{E;nd}(T) \\
\sigma_{A;nd} + \bar{\sigma}_{E;nd}(T) & \sigma_{A;d}^2 + \bar{\sigma}_{E;d}^2(T)
\end{bmatrix}
\right)
\]
Then we can apply the Delta method to obtain the asymptotic variances of $\bar{a}_{\cdot\cdot r;n}/\bar{a}_{\cdot\cdot r;d}$ and $\bar{a}_{\cdot\cdot c;n}/\bar{a}_{\cdot\cdot c;d}$.
One more application of the Delta method then gives~\eqref{eq:theta_hat_var_ratio} for the asymptotic variance of $\hat{\theta}$.
\end{proof}

The standard asymptotically valid level-$\alpha$ CI for $\theta$ based on the central limit theorem for $\hat{\theta}$ takes the form
\[
\hat{\theta} \pm z_{1-\alpha/2} \sqrt{\frac{\hat{V}_{\hat{\theta}}(T)}{N}}
\]
where $z_{\alpha/2}$ is the $1-\alpha/2$ quantile of a standard Gaussian distribution
and $\hat{V}_{\hat{\theta}}(T)$ is a consistent estimator of the asymptotic variance $V_{\hat{\theta}}(T)$ given by~\eqref{eq:theta_hat_var_additive} for additive metrics and~\eqref{eq:theta_hat_var_ratio} for ratio metrics.
Such an estimate can be computed using plug-in consistent estimates of the unknown quantities appearing in~\eqref{eq:theta_hat_var_additive} or~\eqref{eq:theta_hat_var_ratio},
or using jackknife methods~\citep{ma2022indeed}.
It follows that 
\begin{equation}
\label{eq:CI_T_N}
|CI(T;N)| \approx 2z_{1-\alpha/2} \sqrt{\frac{V_{\hat{\theta}}(T)}{N}}.
\end{equation}

\subsection{User-specific temporal correlation}

We are now ready to derive~\eqref{eq:ci_decay_rate} based on~\eqref{eq:CI_T_N},
after making a few additional assumptions and approximations.

\begin{assumption}[No serial correlation]
\label{assump:no_serial_correlation}
For all users $u=1,\ldots,N$, arms $j=r,c$,
and days $t \neq t'$, we have $\cov(e_{utj},e_{ut'j})=0$ for the additive metric and 
\[
\cov\left(
\begin{bmatrix}
e_{utj;n} \\
e_{utj;d}
\end{bmatrix},
\begin{bmatrix}
e_{ut'j;n} \\
e_{ut'j;d}
\end{bmatrix}
\right)
= 
\begin{bmatrix}
0 & 0 \\
0 & 0
\end{bmatrix}
\]
for the ratio metric.
\end{assumption}

Under Assumption~\ref{assump:no_serial_correlation},
the variance of the time-averaged errors
decays as $T^{-1}$:
\begin{equation}
\label{eq:no_serial_correlation}
\bar{\sigma}_E^2(T) = \sigma_E^2/T.
\end{equation} 
Then
ignoring the dependence of $\bar{m}_{\cdot j}$ on the total number of days $T$
(we assume this is not predictable ahead of time,
so for the purposes of experiment sizing it is best to ignore such dependence by default),
by~\eqref{eq:theta_hat_var_additive} we have for additive metrics that
\begin{equation}
\label{eq:se_decay}
\frac{V_{\hat{\theta}}(T)}{V_{\hat{\theta}}(1)} = \frac{1}{T} + \rho \frac{T - 1}{T}
\end{equation}
for 
\begin{equation}
\label{eq:utc_additive}
\rho = \frac{\sigma_A^2}{\sigma_A^2+\sigma_E^2}
\end{equation}
which we define as the user-specific temporal correlation (UTC).
This matches~\eqref{eq:ci_decay_rate}.
A statistical interpretation of the UTC is as the correlation between the metrics $a_{utj}$ and $a_{ut'j}$ for the same user $u$:
\begin{equation}
\label{eq:utc_additive_interpretation}
\cor(a_{utj}, a_{ut'j}) = \frac{\cov(\alpha_{uj}+e_{utj},\alpha_{uj}+e_{ut'j})}{\sqrt{\var(a_{utj})\var(a_{ut'j})}} =  \frac{\sigma_A^2}{\sigma_A^2+\sigma_E^2} = \rho.
\end{equation}

For a ratio metric,
if we ignore the dependence of $\bar{m}_{\cdot j}$ on both $T$ and the arm $j \in \{r,c\}$,
then~\eqref{eq:theta_hat_var_ratio} reveals that~\eqref{eq:se_decay} holds if the UTC is defined by
\[
\rho = \frac{\sigma_{A;n}^2+\sigma_{A;d}^2+2\sigma_{A;nd}}{\sigma_{A;n}^2+\sigma_{A;d}^2+2\sigma_{A;nd}+\sigma_{E;n}^2+\sigma_{E;d}^2 + 2\sigma_{E;nd}}
\]
which gives~\eqref{eq:ci_decay_rate} for the case of ratio metric.
For interpretations' sake,
we verify the following analogue of~\eqref{eq:utc_additive_interpretation} holds for the ratio metric:
\begin{equation}
\label{eq:UTC_ratio}
\cor(a_{utj;n}+a_{utj;d}, a_{ut'j;n}+a_{ut'j;d}) = \rho.
\end{equation}

\section{Adjusting for pre-period data}
\label{sec:prepost}
The analysis to this point has assumed that we do not have access to any pre-experiment data for the users in the experiment.
Using this pre-period data, however, allows us to control for user-specific idiosyncracies,
thereby reducing the variance of the estimate of $\theta$~\citep{deng2013, soriano2017prepost}.
\cite{soriano2017prepost} proposes a Bayesian pre-post estimator that does this in our setting.
We focus on an asymptotically equivalent frequentist variant of this estimator that is easier to analyze.
For ease of exposition, we only consider additive metrics here, though the analysis could be extended to ratio metrics.

For each user $u$ and experiment arm $j \in \{r,c\}$,
we let $X_{uj}$ denote pre-period data which will be correlated with the ``post-period" metrics $a_{utj}$
collected during the experiment. 
Typically (and for the rest of this paper) $X_{uj}$ will consist of user $u$'s observations of the same metric as the experiment data,
averaged over a period of $T_0$ days just prior to the start of the experiment.

Our proposed estimator $\hat{\theta}_{PP}$ is the maximum likelihood estimator for $\theta$ under the following parametric assumption
for all users $u=1,\ldots,N$ and arms $j \in \{r,c\}$:
\begin{equation}
\label{eq:likelihood}
(X_{uj},\bar{a}_{u\cdot j}) \sim \mathcal{N}\left(
\begin{bmatrix}
m_X \\
\bar{m}_{\cdot j}
\end{bmatrix},
\begin{bmatrix}
\sigma_X^2 & \lambda(T) \sigma_X\sigma(T) \\
\lambda(T) \sigma_X\sigma(T) & \sigma^2(T)
\end{bmatrix}
\right).
\end{equation}
Here $\lambda(T) = \cor(X_{uj},\bar{a}_{u\cdot j})$ is the ``pre-post correlation"
and 
\[
\sigma^2(T)=\var(\bar{a}_{u\cdot j}) = \sigma_A^2 + \bar{\sigma}_E^2(T).
\]
Note that $X_{ur}$ and $X_{uc}$ are assumed to be identically distributed,
as they are observed before any treatment has been applied.

\begin{remark}
In cases where the normality assumption in~\eqref{eq:likelihood} is not reasonable,
we can divide the users at random into $Z$ user buckets
and let $X_{uj}$ and $\bar{a}_{u\cdot j}$ correspond to metrics aggregated (summed or averaged) over all users in a given bucket.
Then by the central limit theorem,
the normality assumption becomes reasonable again regardless of the original distribution of the user-level metric.
We note that the estimator $\hat{\theta}$ can be computed using only such aggregated data;
this helps reduce data storage requirements and promotes privacy~\citep{Chamandy2012estimating,soriano2017prepost}.
\end{remark}

\cite{soriano2017prepost} states that non-identifiability of the three mean parameters $m_X$, $\bar{m}_{\cdot r}$, and $\bar{m}_{\cdot c}$
prevents frequentist estimation of a ratio treatment effect.
We note that this is only an issue when we consider the conditional likelihood of the post period data $\bar{a}_{u \cdot j}$ given the pre-period data $X_{uj}$.
By contrast,
if we consider the \textit{joint} likelihood between the pre-period and post-period data,
as we do in~\eqref{eq:likelihood},
then all three parameters are indeed identifiable,
so we have a well-defined maximum likelihood estimate for the ratio treatment effect $\theta$.
Indeed, computing the partial derivatives of the log likelihood
under~\eqref{eq:likelihood}
with our observations of the i.i.d. pairs $(X_{ur}, X_{uc}, \bar{a}_{u \cdot r}, \bar{a}_{u \cdot c})_{u=1}^N$,
and then setting them equal to zero gives the following equations
for the maximum likelihood estimates of the parameters $m_X$, $\bar{m}_{\cdot j}$, $\sigma_X$, $\lambda(T)$, and $\sigma(T)$,
denoted with hats:
\begin{align*}
\hat{m}_X & = \frac{1}{2N} \sum_{j \in \{r,c\}} \sum_{u=1}^N X_{uj} = \bar{X}_{\cdot \cdot} \\
\hat{\bar{m}}_{\cdot j} & = \bar{a}_{\cdot \cdot j} - \frac{\hat{\lambda}(T)\hat{\sigma}(T)}{\hat{\sigma}_X}(\bar{X}_{\cdot j} - \hat{m}_X), \quad j \in \{r,c\} \\
\hat{\theta}_{PP} & = \frac{\hat{\bar{m}}_{\cdot r}}{ \hat{\bar{m}}_{\cdot c}} - 1.
\end{align*}

It follows by the consistency and efficiency results for the MLE, Slutsky’s Theorem, and the Delta Method that
as $N \rightarrow \infty$ we have
\begin{equation}
\label{eq:theta_pp_normality}
\sqrt{N}(\hat{\theta}_{PP}-\hat{\theta}) \stackrel{d}{\rightarrow} \mathcal{N}(0,V_{PP})
\end{equation}
where 
\begin{equation}
\label{eq:theta_pp_var}
V_{PP} = V_{PP}(T) = \sigma^2(T)(\theta+1)^2\left[\left(1-\frac{\lambda(T)^2}{2}\right)(\bar{m}_{\cdot r}^{-2} + \bar{m}_{\cdot c}^{-2}) - \frac{\lambda(T)^2}{\bar{m}_{\cdot r}\bar{m}_{\cdot c}}\right].
\end{equation}
 
As expected, the asymptotic variance $V_{PP}$ 
of the pre-post estimator 
is a decreasing function of the pre-post correlation $\lambda(T)$. 
When there is no treatment effect,
so that $\theta=0$ and $\bar{m}_{\cdot r} = \bar{m}_{\cdot c}$,
we can compute the asymptotic relative efficiency of the original estimator $\hat{\theta}$ to the pre-post estimator $\hat{\theta}_{PP}$ by
\begin{equation}
\label{eq:are}
\frac{V_{PP}(T)}{V_{\hat{\theta}}(T)}=1-\lambda(T)^2
\end{equation}
This implies that the pre-post CI's based on $\hat{\theta}_{PP}$ should have a width that is roughly a fraction $\sqrt{1-\lambda(T)^2}$
of the width of CI's based on $\hat{\theta}$.
We note this is the same as the standard error reduction found by~\citet{deng2013}
when adjusting for pre-period data in estimating an additive treatment effect.

We can obtain an explicit expression for $\lambda(T)$ under Assumptions~\ref{assump:iid_users}-\ref{assump:no_serial_correlation}.
Recall that $X_{uj}$ is computed
as an average of $T_0$ days of the metric for user $u$ in arm $j$ before the start of the experiment.
Then by~\eqref{eq:user_mean_decomposition},
we have
\[
X_{uj} = \frac{1}{T_0} \sum_{t=-T_0}^{-1} (m_{tj} + e_{utj}) + \alpha_{uj}
\]
and $\sigma_X^2=\sigma^2(T_0)$.
Assuming $\rho>0$, we compute
\begin{align}
\lambda(T) & = \frac{\cov(X_{uj},\bar{a}_{u \cdot j})}{\sqrt{\var(X_{uj})}\sqrt{\var(\bar{a}_{u \cdot j})}} \nonumber \\
& = \frac{\var(\alpha_{uj})}{\sqrt{\sigma^2(T_0)}\sqrt{\sigma^2(T)}} \nonumber \\
& = \frac{1}{\sqrt{1+\frac{1-\rho}{\rho T_0}}\sqrt{1+\frac{1-\rho}{\rho T}}} \label{eq:lambda}
\end{align}
where we use the fact that $\sigma^2(T) = \sigma_A^2 + \sigma_E^2/T$ under Assumption~\ref{assump:no_serial_correlation}
and the expression~\eqref{eq:utc_additive} for $\rho$.
If $\rho=0$ then $\lambda(T)=0$ for all $T$, 
and controlling for pre-period data does not offer any variance reduction.
In practice,
it will often increase variance slightly,
since $V_{PP}$
is an asymptotic variance which neglects
the errors in the estimation of the covariance matrix in~\eqref{eq:likelihood}.

From~\eqref{eq:lambda},
we see that the pre-post correlation $\lambda(T)$
is equal to the UTC $\rho$ if $T=T_0=1$.
This makes sense as when $T=T_0=1$,
the pre-post correlation
is the correlation between the metric from a single day before the experiment began,
and the same metric a day after the experiment began.
Since we have implicitly assumed that the start of the experiment
does not change the data generating process,
this is equivalent to the definition of the UTC in~\eqref{eq:utc_additive}.
We also observe from~\eqref{eq:lambda} 
that the pre-post correlation is increasing in both $T$ and $T_0$.
Intuitively, this makes sense,
as when either the pre-period or post-period durations gets longer,
the noise in the errors gets averaged out.

\section{User-day diversion}
\label{sec:user-day}
When a degraded user experience from randomizing the set of users diverted into a given experiment on each day is not a concern,
nor is the detection of learning effects,
an experimenter might consider a ``user-day" diversion.
In that setting,
instead of hashing a user identifier,
one will hash both the user identifier and the current date~\citep{tang2010layer}.
This has the effect of randomly introducing different users into the experiment on different days.
In the case that the proportion of traffic diverted into the experiment is small,
we'd expect (in the absence of network effects)
that data across different days is independent,
since the users on different days are essentially non-overlapping.
This can be conceptualized in~\eqref{eq:user_mean_decomposition}
by absorbing the random effects $\alpha_{ij}$,
into the errors $\epsilon_{itj}$,
i.e. taking $\sigma_A^2=0$ in our derivations above.
Then the UTC $\rho$ is zero as well,
and by~\eqref{eq:se_decay} the variance of $\hat{\theta}$ will decay with duration as $T^{-1}$.
As a 1-day user experiment and a 1-day user-day experiment are identical,
it follows that for all $T \ge 2$,
the variance of $\hat{\theta}$ will be lower for a user-day experiment than for a user experiment of the same size $N$.

However, the ability to control for pre-period data in a user experiment
means that using $\hat{\theta}_{PP}$ in a user experiment
can provide a tighter CI for $\theta$
than using $\hat{\theta}$ in a user-day experiment of the same size $N$.
Indeed we know by~\eqref{eq:are} and~\eqref{eq:lambda} that this will be the case for a 1-day experiment
whenever $\rho>0$ in the user experiment.
As the experiment duration $T$ gets longer,
however, for fixed size $N$ the variance of $\hat{\theta}$ from a user-day experiment decays to 0 at a $T^{-1}$ rate,
while the variance of $\hat{\theta}_{PP}$ in a corresponding user experiment does not decay to 0.
To see this, note that from~\eqref{eq:positive} we know the variance of $\hat{\theta}$ in a user experiment does not go to 0 whenever $\rho>0$.
Yet by~\eqref{eq:lambda},
for fixed $T_0$ the pre-post correlation does not approach 1 as $T$ goes to infinity.
We conclude that for longer duration experiments,
user-day experiments will yield more precise treatment effect estimates than user experiments of the same size.

We provide a simple figure to illustrate this trade-off between user experiments and user-day experiments graphically.
Fig.~\ref{fig:pp_vs_cd2} shows
the theoretical decay in
$V_{PP}(T)$ for a user-day experiment
and in $V_{\hat{\theta}}(T)$ for a user experiment 
as a function of the experiment duration $T$,
when $\rho=0.6$ and $T_0=7$.
We see that for $T \le 11$ days,
the user experiment adjusting for 7 days of pre-period data leads to a estimate of $\theta$ with lower asymptotic variance
than a user-day experiment of the same size.
The reverse is true for experiments longer than 11 days.

\begin{figure}
    \centering
    \includegraphics[width=0.8\linewidth]{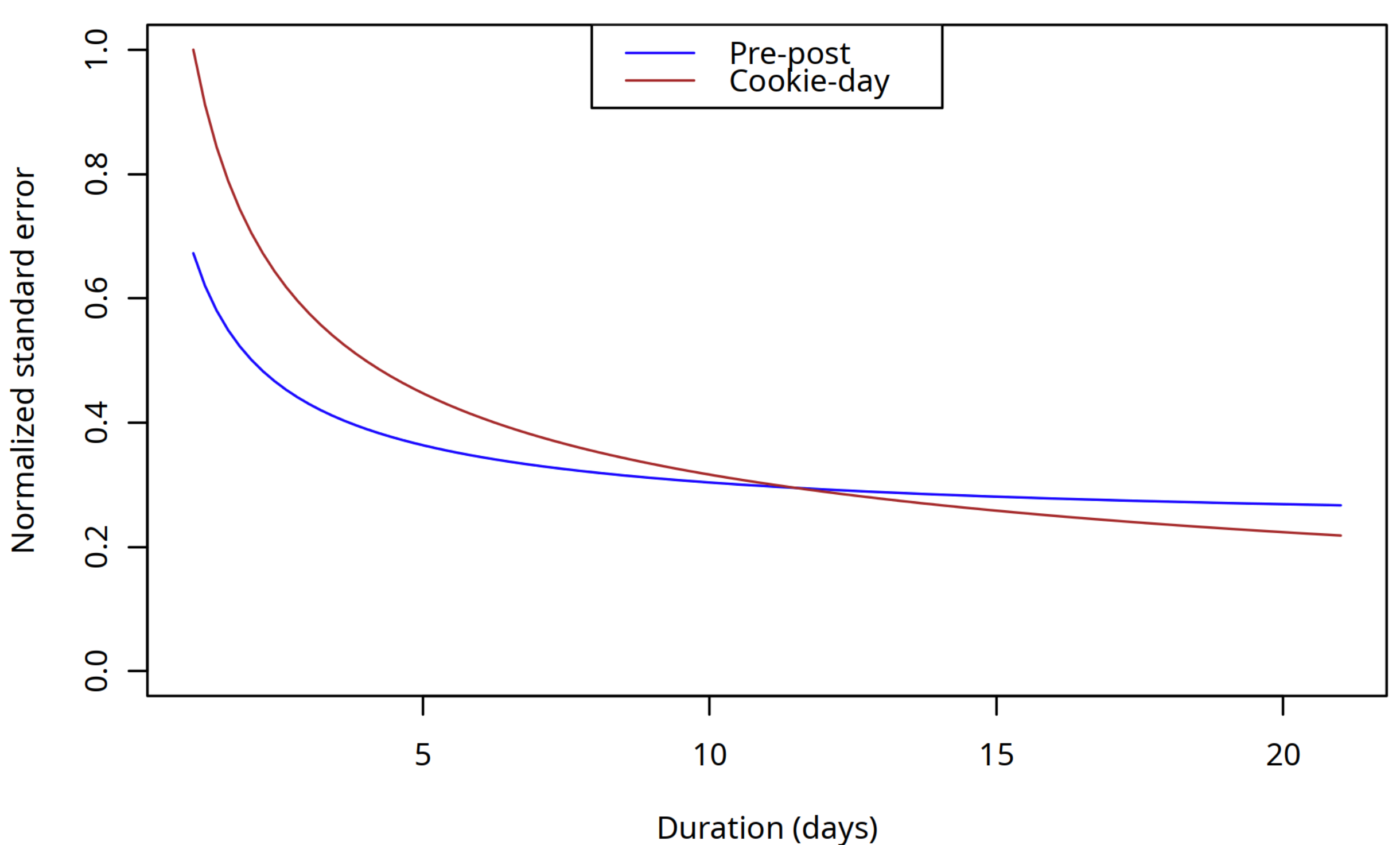}
    \caption{The asymptotic standard error $\sqrt{V_{PP}(T)}$ of the pre-post estimator $\hat{\theta}_{PP}$ as a function of experiment duration $T$ is given in blue.
    The asymptotic standard error $\sqrt{V_{\hat{\theta}}(T)}$ of the post-period estimator $\hat{\theta}$ in a user-day experiment of size $N$ as a function of $T$ is given in brown. The pre-post estimator assumes $T_0=7$ days of pre-period data are available,
    and that the UTC $\rho$ is equal to 0.6 in the user experiment.}
    \label{fig:pp_vs_cd2}
\end{figure}

\section{Estimation of UTC and numerical results}
\label{sec:estimation}
We now describe two ways to estimate the UTC parameter $\rho$
so that~\eqref{eq:ci_decay_rate} can be used to set the duration of a planned experiment.
The first method requires observations $a_{utj}^{(m)}$ of the same metric as the planned experiment from previous experiments
$m=1, 2,\ldots,M$.
% we have access to observations .
% corresponding to $A_{itj}$ for the planned experiment.
% Similarly, if we are trying to estimate variance for the ratio metric,
% we assume access to the cookie bucket aggregated metrics $N_{itj}^{(m)}$ and $D_{itj}^{(m)}$,
% corresponding to $N_{itj}$ and $D_{itj}$, respectively,
% for the present experiment.
We can obtain an estimate $\hat{\rho}^{(m)}$ from experiment $m$ using the sample correlations between $a_{utj}^{(m)}$ and $a_{ut'j}^{(m)}$ for $t \neq t'$
across all users in both arms.
Specifically,
we propose averaging the sample correlations over all pairs $(t,t')$ with $t < t'$.
In practice, if we expect Assumption~\ref{assump:no_serial_correlation} to be somewhat violated,
we might choose to only use pairs $(t,t')$ where $t'-t$ is below a certain lag.
This will typically lead to conservative (upward biased) estimates for $V_{\hat{\theta}}$,
particularly for larger $T$,
since typically we'd expect the serial correlation to decrease with the lag $t'-t$.
To obtain a single UTC estimate $\hat{\rho}$,
we can take an average the $\hat{\rho}^{(m)}$ across experiments
(possibly weighted proportional to their sample sizes, to reduce variance).

Alternatively, if the user-level data is not available
but we have the daily CI widths from past experiments, we can also estimate $\rho$ using equation~\eqref{eq:ci_decay_rate}.
Specifically, we can solve for $\rho$ given the CI widths at the end of any two distinct days in a given experiment.
To obtain our final estimate $\hat{\rho}$,
we simply average the estimates over all the different experiments and day pairs.
In practice, we find this method performs as well as using the raw data from previous experiments.
It is more attractive from an implementation standpoint as it does not require storage of any of the raw metric information from those experiments.

\begin{figure}
    \centering
    \includegraphics[width=0.8\linewidth]{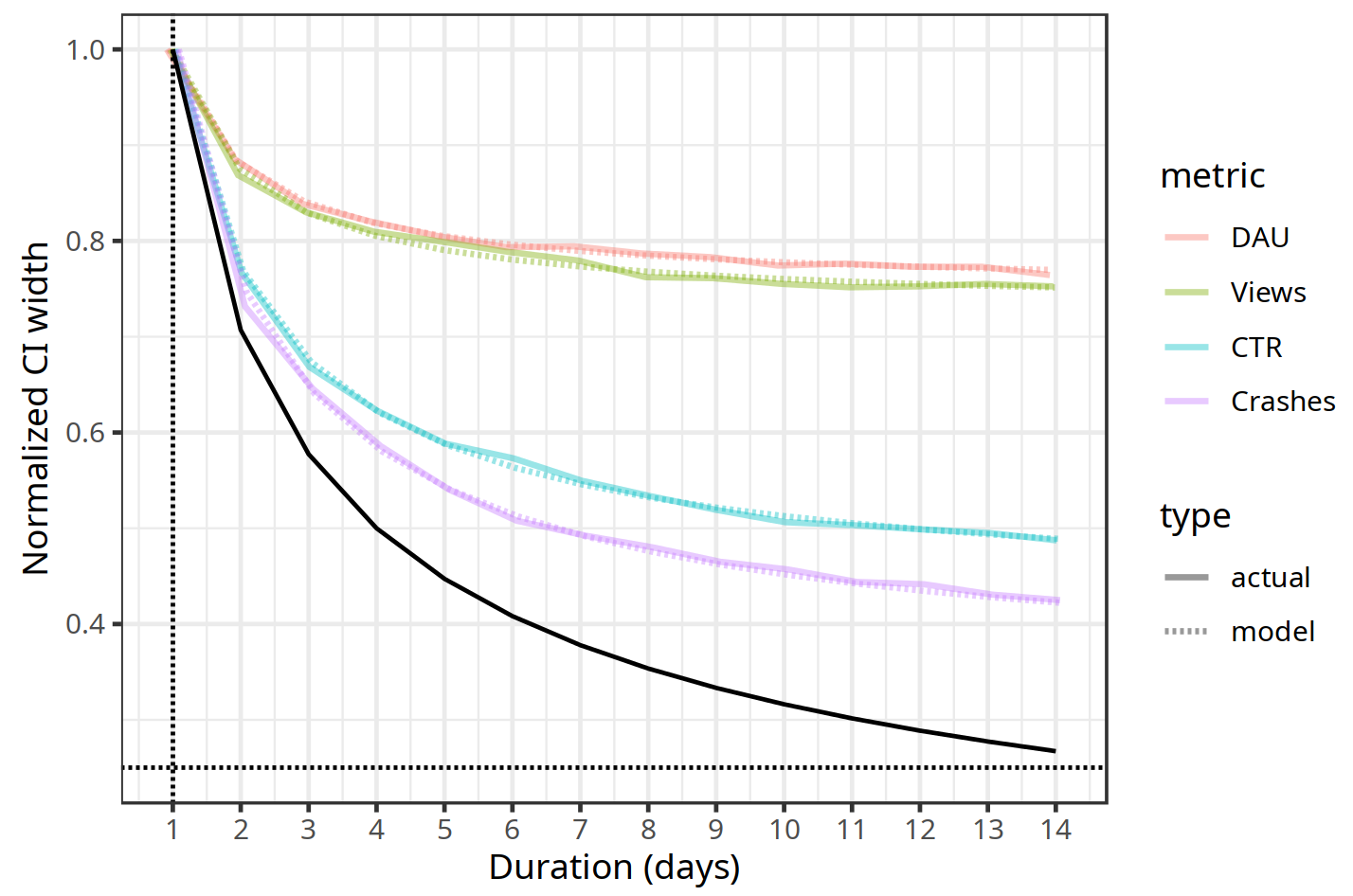}
    \caption{For each metric, the dotted line represents how we predict CI width (normalized by the day 1 CI width) would change with experiment duration $T$ based on our model, and the solid line represents the observed CI width change in actual experiments. The black solid curve describes the $1/\sqrt{T}$ decay rate.}
    \label{fig:actual_vs_fitted}
\end{figure}

We carry out this second estimation procedure for experiments on the four YouTube metrics described above: DAU, Views, CTR and Crashes,
and show that~\eqref{eq:ci_decay_rate} with the estimated UTC parameter $\rho$ closely matches the empirical decay of CI widths in a held-out set of experiments.
The results are presented in Fig. \ref{fig:actual_vs_fitted}. 
Specifically, we used 100 past experiments to generate an estimate $\hat{\rho}$ of the UTC $\rho$ for each metric.
The dotted curves are the normalized CI widths predicted by~\eqref{eq:ci_decay_rate} with the estimate $\hat{\rho}$ plugged in.
We averaged normalized CI widths over 400 separate experiments to give an empirical illustration of how CI width decays over time (solid curves). 

Evidently,
the modeled curves overlap very well with the empirical curves obtained from the CIs in actual experiments. 
In other words,
for each metric the formula~\eqref{eq:ci_decay_rate} models the empirical behavior well.
The simplicity of~\eqref{eq:ci_decay_rate} is useful for building the corresponding experimentation platforms and tools since for each metric we only need to save one parameter to enable planning of experiment duration. 

% We already discussed this above, I don't think we need to say it again.
% Fig.~\ref{fig:actual_vs_fitted} also shows wide variations in the UTC across metrics. Depending on the metric, running the experiment longer might or might not reduce the variance much.
% It seems that for user engagement metrics like Views, within a week the CI width appears to nearly level out due to a high UTC. 
% For metrics like this, running the experiment for a longer duration would thus only be useful for studying long term effects or learning effects. On the other hand, there are metrics like crashes and CTR where running the experiments longer can evidently give much less noisy estimates, due to a small UTC.

\section{Summary and discussion}
We have derived simple formulas~\eqref{eq:ci_decay_rate} and~\eqref{eq:ci_decay_rate_PP} that accurately predict the rate of decay of CI widths in large scale A/B tests in YouTube experiments.
The formula can be easily implemented as a tool or a feature in any experiment analysis portal. 
Here at YouTube, we have found such a tool very helpful in planning for the duration of experiments, stopping unnecessarily long experiments for resource saving.
The analysis of Sections~\ref{sec:prepost} and~\ref{sec:user-day} has also enabled us to more wisely select the experiment type (user experiment with pre-period correction vs. user-day experiment) for further improving statistical power.

In practice, power is not the only consideration in setting the experiment duration $T$.
A common practice is to always run an experiment for at least a full week to allow one to estimate potential day-of-the-week effects. Another common practice is to run the experiment very long (e.g. multiple months) to study the long term effects~\citep{Hohnhold2015focus}.
An interesting direction for future study would be to enhance~\eqref{eq:ci_decay_rate} to incorporate such objectives.

\bibliographystyle{apalike}
\bibliography{ref}
\end{document}